\documentclass[10pt,a4paper,envcountsame]{llncs}

\usepackage[english]{babel}
\usepackage{amsmath,amssymb,amsfonts,amsbsy}
\usepackage{tabularx,subfigure}
\usepackage{booktabs}
\usepackage{array}
\usepackage{multirow}
\usepackage{float}
\usepackage{footnote}
\usepackage{subfigure}
\usepackage{xspace}
\usepackage{url}

\usepackage[protrusion=true,expansion=true]{microtype}

\newcommand{\N}{\mathbb{N}}
\renewcommand{\H}{H}
\newcommand{\NE}{\ensuremath\text{N}\xspace}
\newcommand{\OPT}{\ensuremath\text{O}\xspace}

\newcommand{\PoS}{\text{PoS}}
\newcommand{\PoA}{\text{PoA}}
\newcommand{\POPoS}{\text{POPoS}}
\newcommand{\POPoA}{\text{POPoA}}
\newcommand{\cost}{\text{cost}}
\newcommand{\game}

\usepackage[draft,marginclue,footnote]{fixme}

\usepackage{graphicx}

\newcommand{\keywords}[1]{\par\addvspace\baselineskip\noindent\keywordname\enspace\ignorespaces#1}

\title{\mbox{Improving the $H_k$-Bound on the Price of Stability} in Undirected
Shapley Network Design Games}

\author{Yann Disser \inst{1} \and Andreas Emil Feldmann \inst{2} \and 
Max Klimm \inst{1} \and Mat\'u\v s Mihal\'{a}k \inst{3}}

\institute{
Institut f\"ur Mathematik, Technische Universit\"at Berlin, Germany
\email{\{disser,klimm\}@math.tu-berlin.de}
\and 
Combinatorics and Optimization Department,
University of Waterloo, Canada
\email{andreas.feldmann@uwaterloo.ca}
\and Institute of Theoretical Computer Science,
ETH Zurich, Switzerland
\email{matus.mihalak@inf.ethz.ch}
}

\date{}

\begin{document}
\maketitle

\begin{abstract}
  %
  %
  %
  In this paper we show that the price of stability of Shapley network design
  games on undirected graphs with $k$ players is at most
  $\smash{\frac{k^3(k+1)/2-k^2}{1+k^3(k+1)/2-k^2}\H_k}$ $\smash{= \bigl(1 -
  \Theta(1/k^4)\bigr)\H_k}$, where~$\H_k$ denotes the $k$-th harmonic number.
  This improves on the known upper bound of~$\H_k$, which is also valid for
  directed graphs but for these, in contrast, is tight. Hence, we give the first
  non-trivial upper bound on the price of stability for undirected Shapley
  network design games that is valid for an arbitrary number of players.
  Our bound is proved by analyzing the price of stability restricted to Nash
  equilibria that minimize the potential function of the game. 
  We also present a game with $k=3$ players in which such a restricted price of
  stability is $1.634$. This shows that the analysis of Bil\`o and
  Bove 
  (Journal of Interconnection Networks, Volume 12, 2011)
  is tight. In addition, we give an example for three
  players that improves the lower bound on the (unrestricted) price of stability
  to $1.571$.
\end{abstract}

\section{Introduction}
\looseness=-1

Infrastructure networks are the lifelines of our civilization.
Through generations 
a tremendous effort has been undertaken to cover the earth's surface with
irrigation 
canal systems, sewage lines, road networks, railways, and -- more recently --
data 
networks. Some of these infrastructures are initiated and planned 
by a central authority that designs the network and decides on its topology and 
dimension. Many networks, however, arise as an outcome of actions of selfish
individuals who 
are motivated by their own connectivity requirements rather than by optimizing
the 
overall network design. A prominent example of the latter phenomenon is the rise
of 
the Internet. 
In order to quantify the efficiency of networks, it is crucial to understand 
the processes that govern their formation.
Anshelevich et al.\,\cite{Anshelevich+etal/2008} proposed a particularly
elegant model for such processes, which is now known as the Shapley network
design
game or the network design game with fair cost allocation (for an overview of
other
models for network formation, see \cite{Tardos+Wexler/2007}).

%
%

%
The Shapley network design game is played by $k$ players $1,2,\ldots,k$ on a
graph $G=(V,E)$ with positive edge-costs $c_e\in\mathbb{N}$. Each player $i$ has
associated with it a source-target pair $s_i,t_i \in V$ of vertices that she
needs to connect with a simple path in $G$. The choice of such a path is called
a \emph{strategy} of the player, and a collection consisting of one strategy for
each
player is called a \emph{strategy profile}. 
The cost $c_e$ of every edge $e$ is shared equally among the players 
using it. Each player $i$ aims at choosing a path of smallest possible 
(individual) cost to herself. This cost is defined as the sum of the cost shares
for player~$i$ along the path.
Players are selfish in that they only care about their own costs. In
particular, 
they do not care about the social cost, defined as the sum of all players'
individual
costs and denoted by $\text{cost}(P)$ for a strategy profile $P$.

A \emph{Nash equilibrium} of a Shapley network design game is a strategy profile
in which no player $i$ can switch to an $s_i$-$t_i$ path that yields her a
smaller individual cost.
To quantify the effect of the selfish behavior, it is natural to compare the 
social cost of a Nash equilibrium of the game with the smallest social cost 
among all possible strategy
profiles~\cite{Anshelevich+etal/2008,Tardos+Wexler/2007}.
Several quantifications of selfish behavior have been studied, based on
whether we restrict ourselves to a specific set of Nash equilibria, and whether
we 
compare the worst or best such equilibrium in terms of social cost.
In this paper, we adopt the notion of the \emph{price of stability}, introduced
by Anshelevich et al.\,\cite{Anshelevich+etal/2008}.
Denoting by $\mathcal{N}$ the set of all Nash equilibria and by $\OPT$ a
strategy
profile that minimizes the social cost of a game, the price of
stability of the game is defined as the ratio~$\min_{\NE \in \mathcal{N}}
\text{cost}(\NE)/\text{cost}(\OPT)$.

Anshelevich et al.\,\cite{Anshelevich+etal/2008} observed that Shapley network
design games always have a Nash equilibrium by showing that they belong to the
class of congestion games. For these games, the existence of a Nash equilibrium
is always guaranteed, as shown by Rosenthal~\cite{Rosenthal/1973}. Rosenthal's
existence proof relies on a potential function argument. That is, he showed that
there exists a function $\Phi$ that maps strategy profiles to real numbers and
has the property that if any one player changes her strategy unilaterally,
then the value of $\Phi$ changes by the exact same value as the cost of the
player.
This observation, together with the finiteness of the space of all strategy
profiles, implies the existence of a Nash equilibrium. In particular, any
\emph{potential minimum}, i.e., a strategy profile that globally minimizes the
potential function, is a Nash equilibrium. 
The potential function of a game is unique up to an additive constant (see
Monderer and Shapley~\cite{Monderer+Shapley/1996}).
Using the special form of the potential function for Shapley network design
games, Anshelevich et al.\,\cite{Anshelevich+etal/2008} showed that the price of
stability of any game is at most $\H_k=\sum_{i=1}^k\frac{1}{i}$, the $k$-th
\emph{harmonic number} (which is of order~$\log k$).
This upper bound is tight for games played on directed graphs. That is, there
are Shapley 
network design games on directed graphs~\cite{Anshelevich+etal/2008} for which
the price of stability is arbitrarily close to $\H_k$.

The situation is different for \emph{undirected} Shapley network design games, 
i.e., games played on undirected graphs. As the same potential arguments remain
valid, the price of stability of any game is still at most $\H_k$. Yet, the
largest known price of stability (asymptotically) is a constant, more precisely
$348/155 \approx 2.245$ (see Bil\`{o} et al.\,\cite{Bilo+etal/2010}).
This leaves the question of the worst-possible price of stability in undirected 
Shapley network design games with $k$ players wide open.
Remarkably, the largest known price of stability, as provided by Bil\`{o} et
al., does not come from a simple example, but from a complicated construction. 
Previously known worst-case games had a price of stability of $4/3 \approx
1.333$~\cite{Anshelevich+etal/2008}, $12/7 \approx 1.714$~\cite{Fiat+etal/2006},
and $42/23 \approx 1.8261$~\cite{Christodoulou+etal/2009}.
Despite numerous attempts 
\cite{Bilo+Bove/2011,Anshelevich+etal/2008,Bilo+etal/2010,Fiat+etal/2006,Christodoulou+etal/2009,Li/2009}
to narrow the gap of the bounds on the price of stability, there has been little
progress in terms of numerical results.
It is generally believed that the price of stability is smaller than $H_k$,
and we confirm this belief in this paper. For small values of $k$ some smaller
upper bounds are known. For $k=2$ players, the price of stability is
at most $4/3 < H_2=3/2$ and this is 
tight~\cite{Anshelevich+etal/2008,Christodoulou+etal/2009}. Bil\`{o} and
Bove~\cite{Bilo+Bove/2011} analyzed the case of $k=3$ players and showed that
the price of stability of any such game is at most $1.634<\H_3=1.83\bar{3}$. For
this case, however, a considerable gap remains, as the worst example known has a
price of stability of $74/48 \approx 1.542$~\cite{Christodoulou+etal/2009}.
Thus, already for $k=3$ players, the exact worst-case price of stability is
unknown. 

For several special cases, one can derive better upper bounds on the price of
stability. If all players share the same terminal then the price of stability is
at most $O(\log k/\log \log k)$~\cite{Li/2009}. If in addition every vertex is
the source of at least one player, then the price of stability further degrades
to $O(\log \log k)$~\cite{Fiat+etal/2006}.

Many of the mentioned upper bounds are not only valid for the best Nash
equilibrium of a game, but also for a very specific one -- the potential
minimum. Potential minima have desirable stability properties. For
example, they are reached by certain learning dynamics for players that do not
always play rationally (see Blume~\cite{Blume/1995}). 
This motivates to explicitly study the ratio between the cost of a potential
minimum and that of a profile minimizing the social cost -- a \emph{social
optimum}.
To stress the described stability
properties of potential minima, Asadpour and Saberi~\cite{Asadpour+Saberi/2009}
called this ratio the \emph{inefficiency ratio of stable equilibria}. 
Kawase and Makino~\cite{Kawase+Makino/2012} called the very same ratio the
\emph{potential-optimal price of anarchy}. They also define the
\emph{potential-optimal price of stability} of a game in the obvious way as the
ratio between the cost of a best potential minimum and that of a social
optimum.
%
They prove that the potential-optimal price of anarchy of undirected Shapley
network design games is at most $O(\sqrt{\log k})$ for the special case where
all players share the same terminal node, and where every vertex is the source
of at least one player. They give a construction of a game with
potential-optimal price of anarchy $\Omega(\sqrt{\log \log k})$.

\subsection*{Our Contribution}
\looseness=-1

Our main result shows that the price of stability in undirected
Shapley network design games is at most
$\frac{k^3(k+1)/2-k^2}{1+k^3(k+1)/2-k^2}\H_k=(1-\Theta(1/k^4))\H_{k}$. Thus, we
provide the first general upper bound that shows that the price of stability 
for $k$ players is strictly smaller than $H_k$. 
To prove this upper bound, we generalize the techniques of Christodoulou
et al.\,\cite{Christodoulou+etal/2009} to any number of players. In short,
similar to Christodoulou et al., we obtain a set of inequalities relating the
cost of any Nash equilibrium to the cost of a social optimum. We then combine
these in a non-trivial way to obtain the claimed upper bound,
additionally assuming that the Nash equilibrium has a smaller potential than the
social optimum. 
Interestingly, the resulting upper bound is tight for the case of $k=2$ players.

As an additional contribution, we provide an example of a game with $k=3$
players in which the potential-optimal price of stability is $1.634$. Thus, we
show that the upper bound on the potential-optimal price of anarchy given by
Bil\`{o} and Bove~\cite{Bilo+Bove/2011} is
tight. This result implies that for three players the upper bound
on the price of stability cannot be further improved via potential-minimizers.
This is in contrast to the directed case, for which a simple inequality
relates the cost of the potential minimum to that of a social optimum, giving a
tight bound on the price of stability. We believe that this observation provides
an insight as to why the undirected case is much harder to tackle than the
directed one.

We note that in our tight example for the potential-optimal price of
stability/anarchy, the social optimum is also a Nash equilibrium, and thus the
example provides no new lower bounds on the price of stability.
Our third contribution however is a new lower bound on the price of stability.
We provide an example of a game with three players and price of stability
$1.571$, which improves on the previous best lower bound of $74/48 \approx
1.542$~\cite{Christodoulou+etal/2009}.

%
\section{Problem Definition and Preliminaries}
\label{sec:definition}
\looseness=-1

Let $G=(V,E)$ be an undirected graph with a positive \emph{cost} $c_e>0$
for every edge $e\in E$.
The \emph{Shapley network design game} is a strategic game of $k$ players. Every
player $i\in\{1,\ldots,k\}$ has a dedicated pair of vertices $s_i,t_i\in V$,
that we call her \emph{source} and \emph{target}, respectively.
The \emph{strategy space} of player $i$ is the collection $\mathcal{P}_i$ of all
paths $P_i\subseteq E$ between $s_i$ and $t_i$. Every such path $P_i$ is called 
a \emph{strategy}.
A \emph{strategy profile} $P$ is a tuple $(P_1,\ldots,P_k)$ of $k$ strategies,
$P_i\in\mathcal{P}_i$. 
Given a strategy profile $P$, we say that player $i$ \emph{plays} strategy
$P_i$ in $P$.
The \emph{cost to player} $i$ in a strategy profile $P=(P_1,\ldots,P_k)$ is 
$$
  \text{cost}_i(P) = \sum_{e\in P_i} \frac{c_e}{k_e},
$$
where $k_e$ denotes the number of paths $P_i$ in $P$ such that $e\in P_i$. That
is, $k_e$ is the number of players that use edge $e$.
The goal of every player is to minimize her cost.
A \emph{Nash equilibrium} is a strategy profile $\NE = (\NE_1,\ldots,\NE_k)$,
$\NE_i \in \mathcal{P}_i$, such that no player $i$ can improve her cost by
playing a different strategy. That is, for every~$i$ and every 
$\NE_i'\in \mathcal{P}_i$, it holds that 
$\text{cost}_i(\NE) \leq \text{cost}_i(N_i',N_{-i})$, where 
$(N_i',N_{-i})$ is a shorthand for 
$(\NE_1,\ldots,\NE_{i-1},\NE_i',\NE_{i+1},\ldots,\NE_k)$.
With a slight abuse of terminology we will identify the game played on the graph
$G$ with the graph itself and we write $\mathcal{N}(G)$ to denote the set of
Nash equilibria of $G$.

Observe that the edges of any strategy profile $P$ induce a graph $(V,\cup_i
P_i)$, which we call the \emph{underlying network}. We denote the edge set 
of this graph by~$E(P)$.
The \emph{social cost}, or simply the \emph{cost} of a strategy profile $P$,
denoted as~$\text{cost}(P)$, is the sum of the players' individual costs.
Observe that the social cost is equal to the total cost of the edges in the
played strategies, i.e.,
\begin{displaymath}
  \text{cost}(P) = \sum_{i=1}^k\text{cost}_i(P) = \sum_{e\in E(P)} c_e.
\end{displaymath}

A strategy profile that minimizes the social cost is called the \emph{social
optimum}.
The \emph{price of stability} of a game $G$, denoted by $\PoS(G)$, is defined as
the cost of the best Nash equilibrium of $G$ divided by the cost of a social
optimum $\OPT(G)$ of~$G$. That is, $\PoS(G) = \min_{\NE \in \mathcal{N}(G)}
\cost(\NE) / \cost(\OPT(G))$.
The \emph{price of anarchy} of~$G$, $\PoA(G)$ for short, is obtained by
replacing $\min$ by $\max$ in this definition.
%

A Shapley network design game is a potential
game~\cite{Anshelevich+etal/2008,Rosenthal/1973}. That is, there is a function
$\Phi:\mathcal{P}_1 \times \ldots \times \mathcal{P}_k\rightarrow \mathbb{R}$
such that, for every strategy profile $P$, whenever any player $i$ changes her
strategy from $P_i$ to $P_i'$, then
$\text{cost}_i(P) - \text{cost}_i(P_i',P_{-i}) = \Phi(P) - \Phi(P_i',P_{-i})$.
Rosenthal's (exact) potential function has the form
\begin{displaymath}
  \Phi(P) = \sum_{e\in E(P)} \sum_{i=1}^{k_e}
  \frac{c_e}{i}=\sum_{e\in E(P)}\H_{k_e}\cdot c_e\text{,}
\end{displaymath}
where $\H_j$ denotes the $j$-th Harmonic number $\sum_{i=1}^{j}\frac{1}{i}$.
The potential function is unique up to an additive
constant~\cite{Monderer+Shapley/1996}.

Motivated by the particular stability properties of potential minima, Kawase and
Makino~\cite{Kawase+Makino/2012} introduced two notions to quantify the
inefficiency of potential minimizers. For a game $G$ let $\mathcal{F}(G)$ denote
the set of potential minimizers of $G$, i.e., strategy profiles of $G$ that
minimize the potential function of $G$. The \emph{potential-optimal price of
stability} of $G$ is then defined as $\POPoS(G) = \min_{\NE \in \mathcal{F}(G)}
\cost(\NE) / \cost(\OPT(G))$ and the \emph{potential-optimal price of anarchy}
is
defined as $\POPoA(G) = \max_{\NE \in \mathcal{F}(G)} \cost(\NE) /
\cost(\OPT(G))$.
Since $\mathcal{F}(G) \subseteq \mathcal{N}(G)$, clearly, for any game $G$,
$$\PoS(G) \leq \POPoS(G) \leq \POPoA(G) \leq \PoA(G).$$

For a fixed number of players $k \geq 2$, we are interested to bound the
\emph{worst-case} price of stability of games with $k$ players. For a formal
definition, let $\mathcal{G}(k)$ denote the set of all games with $k$ players.
The price of stability of undirected Shapley network design games with $k$
players is defined as $\PoS(k) = \sup_{G \in \mathcal{G}(k)} \PoS(G)$.
$\POPoS(k)$ and $\POPoA(k)$ are defined analogously. 

Using Rosenthal's potential function $\Phi$, we can bound the potential-optimal
price of anarchy (and, thus, the price of stability) from above by $\H_k$ as
follows~(cf.~\cite{Anshelevich+etal/2008}).
Let $\OPT$ be a social optimum. For a potential minimum $\NE$, we have
$\Phi(\NE)\leq\Phi(\OPT)$. Using this together with $\text{cost}(\NE) \leq
\Phi(\NE)$ and $\Phi(\OPT)\leq H_k\cdot \cost(\OPT)$ we obtain
$\text{cost}(\NE)\leq H_k\cdot \text{cost}(\OPT)$, as claimed.

%

%
We define $\smash{\NE^i}$ and $\smash{\OPT^i}$, for $i\in\{1,\ldots,k\}$, to be
the sets of edges of $\NE$ and~$\OPT$ that are used by exactly $i$ players,
respectively. Thus, \smash{$E(\NE) = \bigcup_{i=1}^k \NE^i$} and \smash{$E(\OPT)
= \bigcup_{i=1}^k \OPT^i$}. For a set of edges $M\subseteq E$, we will denote by
$|M|$ the total cost of the edges in $M$, i.e., $\smash{|M|=\sum_{e\in M} c_e}$.
This allows us to express the value of the potential function for $\NE$ and
$\OPT$ by $\smash{\Phi(\NE)=\sum_{j=1}^{k}\H_{j}|\NE^{j}|}$ and
$\smash{\Phi(\OPT)=\sum_{j=1}^{k} \H_{j}|\OPT^{j}|}$, respectively.

\section{A General Upper Bound}
\label{sec:UB-k-players}

\looseness=-1 
In this section we derive an upper bound on the price of stability of
\emph{any} undirected Shapley network design game. 
The main idea to show our upper bound follows that of Christodoulou et
al.\,\cite{Christodoulou+etal/2009}, which they used for deriving an upper bound
of $33/20=1.65$ for the case of $k=3$ players.
By the definition of Nash equilibria, no player can change her chosen
$s_i$-$t_i$ path in such a profile and thereby improve her cost. 
We will consider a specific change of strategy for each player, which gives us
an inequality that relates the costs of the edges used in the Nash equilibrium
to those used in the social optimum.
We then combine this inequality in a non-trivial way with another that is gained
from the fact that we consider a potential minimum. This allows us to obtain
an upper bound of $\smash{\frac{k^3(k+1)/2-k^2}{1+k^3(k+1)/2-k^2}\H_{k}=(1 -
\Theta(1/k^4))\H_k}$ on the price of stability.
%


Consider a Nash equilibrium $\NE$ and a social optimum $\OPT$ of a given Shapley
network design game with $k$ players. 
In undirected graphs, if $\OPT^k$ is non-empty any player $i$ can use paths in
the underlying network of the optimum~$\OPT$ to connect its terminals to the
source and the target of another player $j$. 
By additionally using the path of player $j$ in the Nash equilibrium~$\NE$,
we obtain a valid strategy for player~$i$. This is the specific alternative
strategy of player $i$ which we use to relate the cost of the Nash equilibrium
to the social optimum. 

We will additionally use the following observation about the structure of social
optima. Observe that any social optimum $\OPT$ forms a forest (because from any
cycle we could remove an edge and thereby decrease the cost).
If $\OPT^k$ is non-empty, i.e., some edges are shared among all players in the
optimum, the edges of $E(\OPT)\setminus\OPT^k$ form two trees such that every
player has one terminal in each of the trees. 
Let~$\OPT^+$ denote the edge set of the larger tree, and~$\OPT^-$ denote the
edge set of the smaller one (measured in terms of the social cost). We have
$\OPT^+\cup\OPT^-=E(\OPT)\setminus\OPT^k$ and, by the definition,
$|\OPT^+|\geq|\OPT^-|$.
%
%
Every tree has a closed walk that visits every vertex at least once and every
edge exactly twice (for example, a depth-first traversal). We consider such a
walk in the tree given by the edges in $\OPT^+$, and use it to order the
players.
Without loss of generality, if $\OPT^{k}\neq\emptyset$, the players are numbered
such that there is a closed walk in $\OPT^+$ that visits the terminals in the
order given by the players' numbers, while using each edge exactly twice. We say
that the players are in \emph{major-tree order}.

Consider the edges of a Nash equilibrium which are not used by all $k$ players.
The following lemma bounds the cost of these edges with respect to the cost of
a social optimum.

\begin{lemma}
Given a game $G$, let $\NE=(\NE_1,\ldots,\NE_k)$ be a Nash equilibrium and
$\OPT=(\OPT_1,\ldots,\OPT_k)$ a social optimum with $\OPT^{k}\neq\emptyset$.
Then,
  \begin{displaymath}
    \sum_{j=1}^{k-1}|\NE^{j}|\leq (k^2(k+1)/2-k)
\sum_{j=1}^{k-1}|\OPT^j|\text{.}
  \end{displaymath}
\label{lem:pot-min bound}
\end{lemma}

\begin{proof}
\looseness=-1
\begin{figure}[t]
\centering
\includegraphics{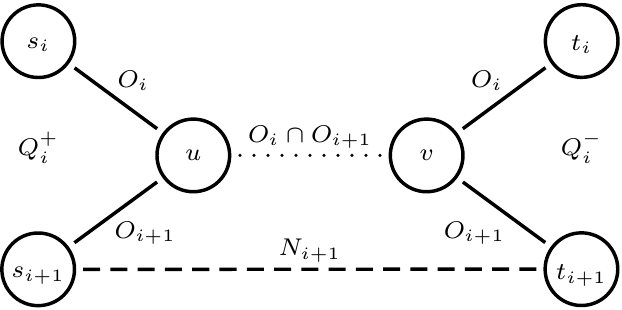}
\caption{Constructing the path $P_i$ with $\OPT_i$, $\OPT_{i+1}$, 
$\NE_{i+1}$ (dashed line), and $Q_i$ (continuous lines). Note
that $\OPT_i\cap\OPT_{i+1}$ (dotted line) is not part of~$P_i$.}
\label{fig:upperbound}
\end{figure}
For every player $i$, we construct an $s_{i}$-$t_{i}$ path $P_i$
(cf.\,Figure~\ref{fig:upperbound}) with the property that every edge on $P_i$ is
either in $\NE_{i+1}$ or not in $\OPT^{k}$. In the following we understand
indices modulo $k$, i.e. $k+1\equiv 1$ and $0\equiv k$. Let $u,v$ be the first
and last vertex on $\OPT_{i}$ that are also on $\OPT_{i+1}$. These vertices are
well defined as $\OPT^k$ is non-empty, and thus $\OPT_{i}\cap\OPT_{i+1}$ is
non-empty. Assume that $u$ lies before $v$ on $\OPT_{i+1}$ (otherwise, exchange
$s_{i+1}$ and $t_{i+1}$ in the following). 
Let $P_{i}$ be the path from $s_{i}$ to $t_{i}$ that
\begin{itemize}
  \item first follows $\OPT_{i}$ until $u$, 
  \item then $\OPT_{i+1}$ (backwards) from $u$ to $s_{i+1}$, 
  \item then $\NE_{i+1}$ from $s_{i+1}$ to $t_{i+1}$, 
  \item then $\OPT_{i+1}$ (backwards) from $t_{i+1}$ to $v$, 
  \item and finally $\OPT_{i}$ from $v$ to $t_{i}$.
\end{itemize}
In case $P_i$ contains cycles, we skip them to obtain a simple path. It is easy
to verify that every edge on $P_i$ either lies on $\NE_{i+1}$ or is not in
$\OPT_{i}\cap\OPT_{i+1}$. Thus, $P_i$ has the desired property. In the
following, let $Q_i$ denote the set of edges of $P_i$ that are also contained
in~$E(\OPT)$.

Since $\NE$ is a Nash equilibrium, player $i$ cannot improve her cost by
choosing path $P_i$. Therefore, $\text{cost}_i(\NE) \leq
\text{cost}_i(P_i,\NE_{-i})$. 
If $k_e$ is the number of players using edge $e$ in $\NE$, this inequality
amounts to 
\begin{displaymath}
  \sum_{e\in\NE_i} \frac{c_e}{k_e} \leq \sum_{e\in P_i\cap\NE_i}
  \frac{c_e}{k_e}+\sum_{e\in P_i\setminus\NE_i} \frac{c_e}{k_e+1}\textrm{.}
\end{displaymath}
By the properties of~$P_i$, the right-hand side of this inequality can be upper
bounded by 
\begin{displaymath}
  \sum_{e\in Q_i} c_e +
  \sum_{e\in\NE_{i+1}\cap\NE_i}\frac{c_e}{k_e} +
  \sum_{e\in\NE_{i+1}\setminus\NE_i}\frac{c_e}{k_e+1}\textrm{.}
\end{displaymath}
By shifting all terms not depending on $Q_i$ to the left-hand side of the
resulting inequality, we get
\begin{equation}
  \label{eq:NEvsQ1}
  \sum_{e\in\NE_i\setminus\NE_{i+1}} \frac{c_e}{k_e} - 
  \sum_{e\in\NE_{i+1}\setminus\NE_i} \frac{c_e}{k_e+1} 
  \leq \sum_{e\in Q_i} c_e\textrm{.}
\end{equation}

We further consider another alternative strategy $\widehat P_i$ of player $i$.
Path $\widehat P_i$ is defined similarly as path $P_i$, but now with respect to
player $i-1$.
%
%
That is, $\widehat P_i$ uses the edges of $\OPT_i$, $\OPT_{i-1}$, and
$\NE_{i-1}$ to connect $s_i$ to $t_i$ and does not contain any edges from
$\OPT_i\cap\OPT_{i-1}$. 
Let $\widehat Q_i$ denote the set of edges of $\widehat P_i$ also contained in
$E(\OPT)$. Using the same arguments as above on $\widehat P_i$ we can derive an
analogous inequality as Inequality~\eqref{eq:NEvsQ1} for edges in $\NE_i$,
$\NE_{i-1}$, and $\widehat Q_i$, i.e, we get 
\begin{equation}
  \label{eq:NEvsQ1hat}
  \sum_{e\in\NE_i\setminus\NE_{i-1}} \frac{c_e}{k_e} - 
  \sum_{e\in\NE_{i-1}\setminus\NE_i} \frac{c_e}{k_e+1} 
  \leq \sum_{e\in \widehat Q_i} c_e\textrm{.}
\end{equation}

Adding Inequalities~\eqref{eq:NEvsQ1} and~\eqref{eq:NEvsQ1hat}, and then summing
over all $i$ gives
\begin{align}
  \label{eq:NEvsQ}
  \sum_{i=1}^k\left(\sum_{e\in\NE_i\setminus\NE_{i+1}} \frac{c_e}{k_e} - 
  \sum_{e\in\NE_{i+1}\setminus\NE_i} \frac{c_e}{k_e+1} +
  \sum_{e\in\NE_i\setminus\NE_{i-1}} \frac{c_e}{k_e} - 
  \sum_{e\in\NE_{i-1}\setminus\NE_i} \frac{c_e}{k_e+1}\right) \nonumber \\ 
  \leq
  \sum_{i=1}^k\left(\sum_{e\in Q_i} c_e + \sum_{e\in \widehat Q_i} c_e\right).
\end{align}

We bound the left-hand side and the right-hand side of
Inequality~\eqref{eq:NEvsQ} separately, starting with the left-hand side. Since
indices are modulo $k$, we may shift the index in the second and the fourth sum
on the left-hand side by, respectively, substracting and adding one to the
index. This lets us combine the first and the fourth sum to
\begin{displaymath}
  \sum_{e\in\NE_i\setminus\NE_{i+1}} \frac{c_e}{k_e} -
  \sum_{e\in\NE_{i}\setminus\NE_{i+1}}\frac{c_e}{k_e+1} =
  \sum_{e\in\NE_i\setminus\NE_{i+1}} \frac{c_e}{k_e(k_e+1)}\textrm{,}
\end{displaymath}
and analogously the second and the third sum to
$\sum_{e\in\NE_i\setminus\NE_{i-1}} \frac{c_e}{k_e(k_e+1)}$, to obtain
\begin{displaymath}
  \sum_{i=1}^k\sum_{e\in\NE_i\setminus\NE_{i+1}}
  \frac{c_e}{k_e(k_e+1)} + \sum_{i=1}^k\sum_{e\in\NE_i\setminus\NE_{i-1}}
  \frac{c_e}{k_e(k_e+1)}\textrm{.}
\end{displaymath}
Each of the two resulting sums counts each edge in $E(\NE)\setminus \NE^k$ at
least once. This is because for any edge $e$ used by at least one player but not
all of them, there is a pair of players with consecutive indices (modulo
$k$) such that $e$ is used in $\NE$ by one of the players but not the other.
Thus, we can lower bound the above term by
\begin{displaymath}
  2\hspace{-.4cm}\sum_{e\in E(\NE)\setminus\NE^k}\frac{c_e}{k_e(k_e+1)} = 
  \sum_{j=1}^{k-1} \frac{2}{j(j+1)} |\NE^j| \geq
  \frac{2}{k(k-1)}\sum_{j=1}^{k-1}|\NE^j|\text{,}
\end{displaymath}
which gives a lower-bound on the left-hand side of Inequality~\eqref{eq:NEvsQ}.

The right-hand side of Inequality~\eqref{eq:NEvsQ} can be bounded by exploiting
the major-tree order of the players. We first only bound the sum depending on
$Q_i$. We denote the two parts of $Q_i$ that lie in the larger and smaller parts
of $E(\OPT)\setminus\OPT^k$ by $Q^+_i=Q_i\cap \OPT^+$ and $Q^-_i=Q_i\cap
\OPT^-$, respectively. Note that, by construction of $P_i$, there are no edges
of $O^k$ in $Q_i$. Thus, we get 
\begin{displaymath}
  \sum_{i=1}^k\sum_{e\in Q_i} c_e = \sum_{i=1}^k\left(\sum_{e\in Q^+_i} c_e +
  \sum_{e\in Q^-_i} c_e \right)\text{.}
\end{displaymath}
By the defining property of the major-tree order, each edge in $\OPT^+$ is
counted exactly twice in the above sum, while each edge of $\OPT^-$ is counted
at most $k$ times. At the same time, the weight of the edges in $\OPT^-$ amounts
to at most half the total weight of $E(\OPT)\setminus\OPT^k$. Hence,
\begin{align*}
  \sum_{i=1}^k\left(\sum_{e\in Q^+_i}\! c_e + \sum_{e\in Q^-_i}\! c_e \right)
  &\leq 2\!\sum_{e\in\OPT^+}\! c_e + k\sum_{e\in\OPT^-}\! c_e 
  = 2\hspace{-.4cm}\sum_{e\in E(\OPT)\setminus\OPT^k}\hspace{-.4cm} c_e + (k-2) \sum_{e\in\OPT^-} c_e \\
  &\leq (k/2+1)\hspace{-.4cm} \sum_{e\in E(\OPT)\setminus\OPT^k} \hspace{-.4cm}
  c_e\ \text{.}
\end{align*}
Analogously, we can derive a corresponding bound for the sum depending on
$\widehat Q_i$, i.e., 
\begin{displaymath}
  \sum_{i=1}^k\sum_{e\in \widehat Q_i} c_e 
  \leq (k/2 + 1) \hspace{-.4cm}\sum_{e\in E(\OPT)\setminus\OPT^k}
  \hspace{-.4cm}c_e \ \text{.}
\end{displaymath}
Since the sum over all costs of edges in
$E(\OPT)\setminus\OPT^k$ is exactly $\sum_{i=1}^{k-1}|\OPT^j|$, we can bound
the right-hand side of Inequality~\eqref{eq:NEvsQ} by
$(k+2)\sum_{i=1}^{k-1}|\OPT^j|$. 

Together, the two derived bounds for the left-hand side and the right-hand side
of Inequality~\eqref{eq:NEvsQ} give the claimed inequality.
\qed
\end{proof}

The following lemma encapsulates some technical calculations that allows to
derive an upper bound on the price of stability using Lemma~\ref{lem:pot-min
bound}.

\begin{lemma}
  \label{lem:beta}
  For a game $G$ with social optimum $\OPT$, let $\NE$ be a Nash equilibrium
  with $\Phi(\NE)\leq\Phi(\OPT)$ and let
  $\beta>0$ be such that $\sum_{j=1}^{k-1}|\NE^j| \leq
  \beta\sum_{j=1}^{k-1}|\OPT^j|$. Then, 
  \begin{displaymath}
    \sum_{j=1}^{k}|\NE^j| \leq \frac{\beta k}{1+\beta k} \H_k
    \cdot\sum_{j=1}^{k}|\OPT^j|\text{.}
  \end{displaymath}
\end{lemma}

\begin{proof}
We take $1<\alpha<\H_{k}$ and compute 
\begin{align*}
  \sum_{j=1}^{k}|\NE^j|  &\leq  \alpha|\NE_{k}|+\sum_{j=1}^{k-1}|\NE^j|
   =  \frac{\alpha}{\H_{k}}\cdot\sum_{j=1}^{k}\H_{j}|\NE^j| +
  \sum_{j=1}^{k-1}\left(1-\alpha\frac{\H_{j}}{\H_{k}} \right)|\NE^j|\\
   &\leq  \frac{\alpha}{\H_{k}}\cdot\sum_{j=1}^{k}\H_{j}|\NE^j| +
  \sum_{j=1}^{k-1}\left(1-\frac{\alpha}{\H_{k}}\right)|\NE^j|\text{.}
\end{align*}
Using that $\Phi(\NE) \leq \Phi(\OPT)$ and the condition of the lemma, we
introduce $\beta$
to get
\begin{eqnarray}
  \sum_{j=1}^{k}|\NE^j| 
  & \leq & \frac{\alpha}{\H_{k}}\cdot\sum_{j=1}^{k}\H_{j}|\OPT^j| +
  \left(1-\frac{\alpha}{\H_{k}}\right)\cdot\beta\sum_{j=1}^{k-1}
  |\OPT^j|\nonumber \\
  & = & \alpha|\OPT^k|+\sum_{j=1}^{k-1}\left[\alpha\frac{\H_{j}}{\H_{k}} +
  \beta\left(1-\frac{\alpha}{\H_{k}}\right)\right] |\OPT^j|\nonumber \\
  & \leq & \alpha|\OPT^k|+\sum_{j=1}^{k-1}\left[\alpha\frac{\H_{k-1}}{\H_{k}} +
  \beta\left(1-\frac{\alpha}{\H_{k}}\right)\right] |\OPT^j|\text{.}
  \label{eq:first}
\end{eqnarray}
For $\alpha=\frac{\beta k}{1+\beta k}\H_{k}$ we have
\begin{eqnarray}
  \alpha\frac{\H_{k-1}}{\H_{k}}+\beta\left(1-\frac{\alpha}{\H_{k}}\right) 
  & = & \frac{\beta k}{1+\beta k}\H_{k-1} + \beta\left(1-\frac{\beta k}{1+\beta
  k}\right)\nonumber \\
  & = & \frac{\beta k}{1+\beta k}\left(\H_{k}-\frac{1}{k}\right) +
  \frac{\beta}{1+\beta k} \nonumber \\
  &=& \alpha. 
  \label{eq:second}
\end{eqnarray}
From (\ref{eq:first}) and (\ref{eq:second}) it follows that 
\begin{displaymath}
  \sum_{j=1}^{k}|\NE^j| \leq \alpha|\OPT^k|+\alpha\sum_{j=1}^{k-1}|\OPT^j| =
  \frac{\beta k}{1+\beta k}\H_{k}\sum_{j=1}^{k}|\OPT^j|\text{,}
\end{displaymath}
which concludes the proof.
\qed\end{proof}

The above lemmas can be put together in order to show the following theorem.

\begin{theorem}\label{thm:general_upper_bound_on_PoS}
The potential-optimal price of anarchy $\POPoA(k)$ for Shapley network design
games with $k \geq 2 $ players is at most
$\frac{k^3(k+1)/2-k^2}{1+k^3(k+1)/2-k^2}\H_{k}.$
\end{theorem}

\begin{proof}
  Let $\NE$ be a potential minimum. If $\OPT^k\neq\emptyset$, we may combine
  Lemma~\ref{lem:pot-min bound} and Lemma~\ref{lem:beta} to obtain
  \begin{displaymath}
    \frac{\mathrm{cost}(\NE)}{\mathrm{cost}(\OPT)} \leq
    \frac{k^3(k+1)/2-k^2}{1+k^3(k+1)/2-k^2}\H_{k}\text{.}
  \end{displaymath}
  If, on the other hand, $\OPT^k=\emptyset$, then obviously also $|\OPT^k|=0$.
We 
  show that then $\mathrm{cost}(\NE)\leq \H_{k-1} \mathrm{cost}(\OPT)$
  (which has been observed before for $k=3$, e.g., by Christodoulou et
  al.\,\cite{Christodoulou+etal/2009}). 
  We can express the potential functions of $\NE$ and $\OPT$ using $\NE^j$ and
  $\OPT^j$ and use the fact that \mbox{$\Phi(\NE)\leq\Phi(\OPT)$} to obtain
  \begin{displaymath}
    \sum_{j=1}^{k}|\NE^j| \leq \sum_{j=1}^{k}\H_{j}|\NE^j| \leq
    \sum_{j=1}^{k}\H_{j}|\OPT^j| = \sum_{j=1}^{k-1}\H_{j}|\OPT^j| \leq
    \H_{k-1}\sum_{j=1}^{k}|\OPT^j|\text{.}
  \end{displaymath}
  Hence, in this case we get 
$$    \frac{\mathrm{cost}(\NE)}{\mathrm{cost}(\OPT)}\leq\H_{k-1}
    <\frac{k^3(k+1)/2-k^2}{1+k^3(k+1)/2-k^2}\H_{k}\text{.}
  $$
  %
  In both cases, the cost of the Nash equilibrium is at most
  $\frac{k^3(k+1)/2-k^2}{1+k^3(k+1)/2-k^2}\H_{k}$ times the cost of the social
optimum.
  \qed
\end{proof}

We obtain the following corollary. Note that the bound is tight
for~$k=2$.

\begin{corollary}
The price of stability $\PoS(k)$ for Shapley network design game with $k \geq 2$
players
is at most $\frac{k^3(k+1)/2-k^2}{1+k^3(k+1)/2-k^2}\H_{k}$.
\end{corollary}


\section{Three-Player Games}

\looseness=-1 
The upper bound of $\frac{k^3(k+1)/2-k^2}{1+k^3(k+1)/2-k^2}\H_{k}$ on the price
of
stability presented in the previous section is valid for an arbitrary number of
players. For the special case of $k=3$ players, it evaluates to $165/92 \approx
1.793$. For this case, however, better bounds are known. Bil\`o and Bove proved
that the price of stability does not exceed $286/175 \approx 1.634$. For the
proof of their result they combine inequalities that are valid for any potential
minimum of the game. Thus, their proof implies that also the potential-optimal
price of anarchy (and thus the potential-optimal price of stability) is at most
$286/175$.
As the main result of this section, we will show that this result is tight. That
is, there is a three-player game such that the cost of the best potential
minimum is $286/175$ times the cost of the social optimum.

\begin{theorem}
  For three players, the potential-optimal price of stability and the
  potential-optimal price of anarchy are $\POPoA(3) \!=\! \POPoS(3) \!=\!
  286/175\!\approx\! 1.634$.
\end{theorem}

\begin{proof}
\looseness=-1 
The upper bound of $286/175$ on the potential-optimal price of anarchy was
proved by Bil\`o and Bove~\cite[Theorem~3.1]{Bilo+Bove/2011}. They derive the
bound for the potential-optimal price of anarchy, but only
explicitly state the implied bound for the (regular) price of stability.

Consider the three player game in Figure~\ref{POPoS_lower_bound_fig}, and let 
$\epsilon>0$ be sufficiently
small. The potential-optimal price of stability of this example approaches
$286/175$ when $\epsilon$ tends to $0$, which establishes tightness of the
upper bound. 
It also shows that the
potential-optimal price of stability and the potential-optimal price of anarchy
coincide for the class of games with three players.

Obviously, any strategy profile in the example has to use at least three edges
to connect all terminal pairs. The three cheapest edges already connect all
terminal pairs and thus constitute the social optimum $\OPT$ of cost
$700+3\epsilon$.
It is easy to verify that the underlying networks of Nash equilibria in the example
do not
contain cycles, since at least one edge of each cycle would be abandoned by all
players. Hence, all Nash equilibria in the example use exactly three edges.
We show that the \emph{unique} potential minimum $\NE$ uses the three edges not
used in~$\OPT$ (note that $\OPT$ itself is a Nash equilibrium). This profile has
both a potential-function value and cost equal to $396+2\cdot374=1144$, since
every edge is used by one player only. In contrast, the social optimum has a
potential-function value of $2\H_2\cdot(209+\epsilon) + \H_3\cdot(282+\epsilon)
> 1144$. 
If the edge $\{t_2,s_3\}$ is used in a Nash equilibrium, the other two edges
have to be used by at least two players each. For profiles other than $\OPT$,
this gives a potential function value of at least 
$\H_2\cdot(209+\epsilon+374) + (282+\epsilon) > 1156$. 
If the edges with cost 396 and $(282+\epsilon)$ are both unused,
all three players use an edge of cost 374 and the resulting potential value is
either $\H_3\cdot 374 + \H_2\cdot(209+\epsilon) + (209+\epsilon) > 1208$ or 
$\H_2\cdot 374 + 374 + (209+\epsilon) > 1144$. 
Equilibria that use one edge each of costs $(209+\epsilon)$, 374, 396 have a
potential function value of $H_2\cdot396+374+(209+\epsilon) > 1177$. And
finally, the
profile using both cheap edges together with the one of cost $396$ has a
potential of $\H_3\cdot 396+2\cdot(209+\epsilon) > 1144$. 
We conclude that the potential minimum is as claimed. For $\epsilon$ tending to
$0$, the ratio between the cost of $\NE$ and~$\OPT$ approaches $286/175$.
\qed
\end{proof}

\begin{figure}[t] 
\centering
\subfigure[\label{POPoS_lower_bound_fig}]{
\includegraphics{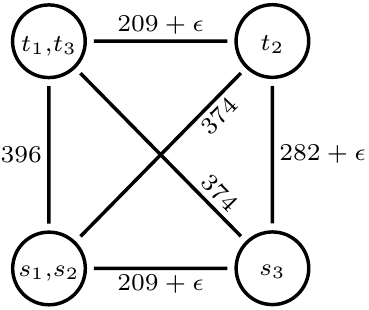}
}
\hspace{1cm}
\subfigure[\label{fig:lb2}]{
\includegraphics{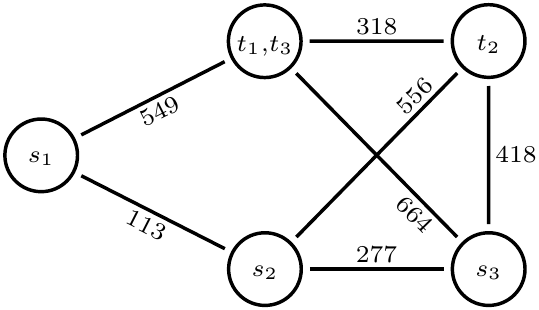}
}
\caption{(a) A three-player game with $\POPoS$ and $\POPoA$ approaching
$286/175 \approx 1.634$ for $\epsilon \to 0$. (b) A three-player game with $\PoS
= 1769/1126 \approx 1.571.$}
\end{figure}

Our result in particular implies that it is impossible to push the upper bound
on the price of stability for three-player Shapley network design games on
undirected networks below $286/175$ by using inequalities that are only valid
for global minima of the potential function.  Note that the example in
Figure~\ref{POPoS_lower_bound_fig} has a price of stability of 1, since its
social optimum is itself a Nash equilibrium.

So far, the best lower bound on the price of stability for three-player games
was
$74/48 \approx 1.542$~\cite{Christodoulou+etal/2009}.
We can slightly improve this bound by presenting a game with three players whose
price of stability is $1769/1126 \approx 1.571$. Consider the network with 5
vertices shown in Figure~\ref{fig:lb2}. By exhaustive enumeration of all
strategy profiles, one can verify that only the strategy profile in which each
player~$i$ uses the edge $(s_i,t_i)$ is a Nash equilibrium. The social optimum
uses all other edges and has a cost of 1126, while the unique Nash equilibrium
has a cost of 1769. This establishes the claimed lower bound on the price of
stability in undirected Shapley network design games with three players.


\section{Conclusions}

We gave an upper bound for the price of stability for an arbitrary number of
players~$k$ in undirected graphs. Our bound is smaller than $\H_k$ for every $k$
and tight for two players. For three players, we showed that the upper bound of
$286/175\approx 1.634$ by Bil\`o and Bove~\cite{Bilo+Bove/2011} is tight for
both the potential-optimal price of stability and anarchy. We also improved the
lower bound to $1769/1126 \approx 1.571$ for the price of stability in this
case.

Asymptotically, a wide gap remains between the upper bound of the price of
stability that is of order $\log k$, and the best known lower bound construction
by Bil\`{o} et al.\,\cite{Bilo+etal/2010} that approaches the constant 
$348/155 \approx 2.245$. It is unclear where the correct answer lies within this
gap, in particular since bounds of $O(\log k/\log \log k)$~\cite{Li/2009} and
$O(\log \log k)$~\cite{Fiat+etal/2006} emerge for restrictions of the problem.
Our bound approaches $\H_k$ for a growing number of players $k$. It would already
be interesting to know whether the price of stability is asymptotically below
$\H_{k-c}$ or $\H_{k/c}$ for some (or even any) positive constant $c\in\N$.

\vspace{8pt}
\noindent\textbf{Acknowledgements.} This work has been partially supported by
the Swiss National Science Foundation (SNF) under the grant number
200021\_143323/1.
We thank Julian Mestre for pointing out a mistake in the proof of
Lemma~\ref{lem:pot-min bound} in an earlier version of this paper.


\bibliographystyle{splncs}
\bibliography{NetworkDesignGames}

\end{document}